\documentclass[11pt,a4paper]{article}
\usepackage[margin=2.5cm]{geometry}
\usepackage{amsfonts,amsmath,latexsym,amssymb,amsthm}
\usepackage[UKenglish]{babel}

\usepackage{graphicx}
\usepackage{xcolor}
\usepackage{xspace}
\usepackage{paralist}
\usepackage{enumitem}
\usepackage{bm}

\usepackage[utf8]{inputenc}
\usepackage{lmodern}
\usepackage{times}
\usepackage{authblk}
\usepackage{hyperref}
\hypersetup{pdftitle={Streaming Algorithms for Bin Packing and Vector Scheduling}}
\hypersetup{pdfauthor={Graham Cormode and Pavel Vesel\'{y}}}

\newcounter{thm}
\theoremstyle{plain}
\newtheorem{theorem}[thm]{Theorem}

\newtheorem{lemma}[thm]{Lemma}
\newtheorem{fact}[thm]{Fact}

\makeatletter
\renewcommand\paragraph{\@startsection{paragraph}{4}{\z@}%
                                      {0.4em}%
                                      {-1em}%
                                      {\normalfont\normalsize\bfseries}}
\makeatother

\makeatletter
\newtheorem*{rep@theorem}{\rep@title}
\newcommand{\newreptheorem}[2]{%
\newenvironment{rep#1}[1]{%
 \def\rep@title{#2 \ref{##1}}%
 \begin{rep@theorem}}%
 {\end{rep@theorem}}}
\makeatother

\newreptheorem{lemma}{Lemma}

\newcommand{\mycase}[1]{\smallskip\noindent{\bf Case #1:}}
\newcommand{\etal}{{\em et al.}}

\def\OPT{\textsf{OPT}\xspace}
\def\ALG{\textsf{ALG}\xspace}
\def\SIZE{\textsf{SIZE}\xspace}

\newcommand{\polylog}{\operatorname{polylog}}
\newcommand{\poly}{\operatorname{poly}}

\newcommand{\half}{\frac{1}{2}}
\let\eps\varepsilon
\newcommand{\oneovereps}{\frac{1}{\eps}}
\newcommand{\oneoverepssquared}{\frac{1}{\eps^2}}

\newcommand{\E}{\mathbb{E}}

\renewcommand{\O}{\mathcal{O}}
\newcommand{\Otilde}{\widetilde{\O}}
\newcommand{\calS}{\mathcal{S}}

\newcommand{\IB}{I_{\mathrm{B}}}
\newcommand{\NB}{N_{\mathrm{B}}}
\newcommand{\IR}{I_{\mathrm{R}}}
\newcommand{\IC}{I_{\mathrm{C}}}

\title{Streaming Algorithms for Bin Packing and Vector Scheduling}

\author{Graham Cormode and Pavel Vesel\'{y}}
\affil{Department of Computer Science, University of Warwick, Coventry, UK.
\texttt{\{G.Cormode,Pavel.Vesely\}@warwick.ac.uk}.}
\date{}

\begin{document}
	
\maketitle

\begin{abstract}
Problems involving the efficient arrangement of simple objects, as captured by bin
packing and makespan scheduling,
are fundamental tasks in combinatorial optimization.
These are well understood in the traditional online and offline cases,
but have been less well-studied when the volume of the input is truly
massive, and cannot even be read into memory.
This is captured by the streaming model of computation, where the aim
is to approximate the cost of the solution in one pass over the data,
using small space. As a result, streaming algorithms produce
concise input summaries that approximately preserve the optimum value.

We design the first efficient streaming algorithms
for these fundamental problems in combinatorial optimization.
For \textsc{Bin Packing}, we provide a streaming asymptotic $1+\eps$-approximation
with $\Otilde\left(\oneovereps\right)$ memory,
where $\Otilde$ hides logarithmic factors.
Moreover, such a space bound is essentially optimal. 
Our algorithm implies a streaming $d+\eps$-approximation for \textsc{Vector Bin Packing}
in $d$ dimensions, running in space $\Otilde\left(\frac{d}{\eps}\right)$.
For the related \textsc{Vector Scheduling} problem, we show how to construct
an input summary in space $\Otilde(d^2\cdot m / \eps^2)$ that preserves
the optimum value up to a factor of $2 - \frac{1}{m} +\eps$,
where $m$ is the number of identical machines.
\end{abstract}

\section{Introduction}

The streaming model captures
many scenarios when we must process very large volumes of data, which cannot
fit into the working memory.
The algorithm makes one or more passes over the data with a limited memory,
but does not have random access to the data.
Thus, it needs to extract a concise summary of the huge input,
which can be used to
approximately answer
the problem under consideration.
The main aim is to provide a good trade-off between the space used 
for processing the input stream (and hence, the summary size) and
the accuracy of the (best possible) answer computed from the summary. 
Other relevant parameters are the time and space needed to make the
estimate, and the number of passes,
which ideally should be equal to one.

While there have been many effective streaming algorithms designed for
a range of problems in statistics, optimization, and graph algorithms
(see surveys by Muthukrishnan~\cite{muthukrishnan05_data_streams} and
McGregor~\cite{mcGregor14_graph_streams}), there has been little
attention paid to the core problems of packing and scheduling.
These are fundamental abstractions, which form the basis of many
generalizations and extensions~\cite{coffman13,christensen17_survey_multidim_BP}.
In this work, we present the first efficient algorithms for packing
and scheduling that work in the streaming model.

A first conceptual challenge is to resolve what form of answer is
desirable in this setting.
If items in the input are too many to store, then it is also unfeasible to
require a streaming algorithm to provide an explicit description of
how each item is to be handled.
Rather, our objective is for the algorithm to provide the cost of the
solution, in the form of the number of bins or the duration of the
schedule.
Moreover, many of our algorithms can provide a concise {\em description}
of the solution, which describes in outline how the jobs are treated
in the design.

A second issue is that the problems we consider, even in their
simplest form, are NP-hard.
The additional constraints of streaming computation do not erase the
computational challenge.
In some cases,
our algorithms proceed by adopting and extending known polynomial-time
approximation schemes for the offline versions of the problems,
while in other cases, we come up with new approaches.
The streaming model effectively emphasizes the question of how compactly can the
input be summarized to allow subsequent approximation of the problem
of interest.
Our main results show that in fact the inputs for many of our problems of interest
can be ``compressed'' to very small intermediate descriptions which
suffice to extract near-optimal solutions for the original input.
This implies that they can be solved in scenarios which are storage or
communication constrained. 

We proceed by formalizing the streaming model, after which we 
summarize our results. 
We continue by presenting related work, and contrast with the online
setting. 

\subsection{Problems and Streaming Model}

\paragraph{Bin Packing.}
The \textsc{Bin Packing} problem is
defined as follows:
The input consists of $N$ items
with sizes $s_1, \dots, s_N$ (each between 0 and 1), which need to be packed
into bins of unit capacity. That is, we seek a partition of the set of items $\{1, \dots, N\}$
into subsets $B_1, \dots, B_m$, called bins, such that for any bin $B_i$, it holds that $\sum_{j\in B_i} s_j \le 1$.
The goal is to minimize the number $m$ of bins used.

We also consider the natural generalization to \textsc{Vector Bin Packing},
where the input consists of $d$-dimensional vectors, with the value of each
coordinate between 0 and 1 (i.e., the scalar items $s_i$ are replaced
with vectors $\mathbf{v^i}$).
The vectors need to be packed into $d$-dimensional bins with
unit capacity in each dimension, we thus require
that
$\|\sum_{\mathbf{v}\in B_i} \mathbf{v}\|_\infty \le 1$ (where, the
infinity norm $\|\mathbf{v}\|_\infty = \max_i \mathbf{v}_i$).

\paragraph{Scheduling.}
The \textsc{Makespan Scheduling} problem is
 closely related to \textsc{Bin Packing} but,
instead of filling bins with bounded capacity, we try to balance the loads assigned
to a fixed number of bins.
Now we refer to the input as comprising a set
of \emph{jobs}, with each job $j$ defined by its processing time $p_j$.
Our goal is to assign each job on one of $m$ identical machines
to minimize the \emph{makespan}, which is the maximum load over all machines.

In \textsc{Vector Scheduling},
a job is described not only by its processing time, but also by, say, memory
or bandwidth requirements. 
The input is thus a set of jobs, each job $j$ characterized by a vector $\mathbf{v^j}$.
The goal is to assign each job into one of $m$ identical machines
such that the maximum load over all machines and dimensions is minimized.

\paragraph{Streaming Model.}
In the streaming scenario, the algorithm receives the input as a sequence of items,
called the input stream.
We do not assume that the stream is ordered in any particular way (e.g., randomly
or by item sizes), so our algorithms must work for arbitrarily ordered streams.
The items arrive one by one and upon receiving each item, the algorithm updates its memory
state. A streaming algorithm is required to use space sublinear in the length of the stream, 
ideally just $\polylog(N)$, while it processes the stream. After the last item arrives,
the algorithm computes its estimate of the optimal value, and the space or time used
during this final computation is not restricted.

For many natural optimization problems
outputting some explicit solution of the problem is not possible owing to the memory restriction
(as the algorithm can store only a small subset of items).
Thus the goal is to find a good approximation of the \emph{value} of an offline optimal solution.
Since our model does not assume that item sizes are integers,
we express the space complexity not in bits, but in words (or memory cells),
where each word can store any number from the input; a linear
combination of numbers  from the input; 
or any integer with $\O(\log N)$ bits (for counters, pointers, etc.).

\subsection{Our Results}

\paragraph{Bin packing.}
In Section~\ref{sec:bp}, we present a streaming algorithm for \textsc{Bin Packing}, which outputs
an asymptotic $1+\eps$-approximation of $\OPT$, the optimal number of bins, using 
$\O\left(\frac{1}{\eps} \cdot \log \frac{1}{\eps} \cdot \log \OPT\right)$ memory.
This means that the algorithm uses
at most $(1+\eps)\cdot \OPT + o(\OPT)$ bins, and in our case, the additive $o(\OPT)$ term
is bounded by the space used.
The novelty of our contribution is to combine 
a data structure that approximately tracks 
all quantiles in a numeric stream~\cite{greenwald01_quantile_summaries}
with techniques for approximation schemes~\cite{fernandez81,karmarkar82}.
We show that we can improve upon the $\log \OPT$ factor in the space complexity
if randomization is allowed or if item sizes are drawn from a bounded-size set of real numbers.
On the other hand, we argue that our result is close to optimal, up to a factor of $\O\left(\log \oneovereps\right)$,
if item sizes are accessed only by comparisons (including comparisons with some fixed constants).
Thus, one cannot get an estimate with at most $\OPT + o(\OPT)$ bins by
a streaming algorithm, unlike in the offline setting~\cite{hoberg17_logarithmic}.
The hardness emerges from the space complexity of the quantiles
problem in the streaming model. 

For \textsc{Vector Bin Packing}, 
we design a streaming asymptotic $d+\eps$-approximation algorithm
running in space
$\O\left(\frac{d}{\eps} \cdot \log \frac{d}{\eps} \cdot \log \OPT\right)$;
see Section~\ref{app:vbp}.
We remark that if vectors are rounded into a sublinear number of types,
then better than $d$-approximation is not possible~\cite{bansal16_vector_BP}.

\paragraph{Scheduling.}
For \textsc{Makespan Scheduling}, one can obtain a straightforward streaming
$1+\eps$-app\-ro\-xi\-ma\-tion\footnote{Unlike for \textsc{Bin Packing},
an additive constant or even an additive $o(\OPT)$ term
does not help in the definition of the approximation ratio,
since we can scale every number on input by any $\alpha > 0$ and $\OPT$
scales by $\alpha$ as well.
} with space of only $\O(\oneovereps\cdot \log \oneovereps)$
by rounding sizes of suitably large jobs to powers of $1+\eps$ and counting
the total size of small jobs.
In a higher dimension, it is also possible to get a streaming
$1+\eps$-approximation, by the rounding introduced by Bansal \etal~\cite{bansal16_apx_vector_sch_eptas}.
However, the memory required for this algorithm is exponential in $d$,
precisely of size $\O\left(\left(\frac{1}{\eps} \log \frac{d}{\eps}\right)^d\right)$,
and thus only practical when $d$ is a very small constant.
Moreover, such a huge amount of memory is needed even if the number $m$ of machines
(and hence, of big jobs) is small
as the algorithm rounds small jobs into exponentially many types.
See Section~\ref{app:vsch-roudingAlgs} for more details.

In case $m$ and $d$ make this feasible,
we design a new streaming $\left(2-\frac{1}{m}+\eps\right)$-app\-ro\-xi\-ma\-tion
with $\O\left(\oneoverepssquared \cdot d^2\cdot m\cdot \log \frac{d}{\eps}\right)$ memory,
which implies a $2$-approximation streaming algorithm running in space $\O(d^2\cdot m^3\cdot \log dm)$.
We thus obtain a much better approximation than for \textsc{Vector Bin Packing}
with a reasonable amount of memory
(although to compute the actual makespan from our input summary,
it takes time doubly exponential in $d$~\cite{bansal16_apx_vector_sch_eptas}).
Our algorithm is not based on rounding, as in the aforementioned algorithms,
but on combining small jobs into containers, and the approximation guarantee 
of this approach is at least $2-\frac{1}{m}$, which we demonstrate by an example.
We describe the algorithm for \textsc{Vector Scheduling} in Section~\ref{sec:vsch}.

\section{Related Work}

We give an overview of related work in offline,
online, and sublinear algorithms, and 
highlight the differences between online and streaming algorithms.
Recent surveys of Christensen \etal~\cite{christensen17_survey_multidim_BP}
and Coffman \etal~\cite{coffman13} have a more comprehensive overview.

\subsection{Bin Packing}

\paragraph{Offline approximation algorithms.}
\textsc{Bin Packing} is an NP-complete problem and indeed it is
NP-hard even to decide
whether two bins are sufficient or at least three bins are necessary.
This follows by a simple reduction from the \textsc{Partition} problem
and presents the strongest inapproximability to date.
Most work in the offline model focused on providing \emph{asymptotic} $R$-approximation algorithms,
which use at most $R\cdot \OPT + o(\OPT)$ bins.
In the following, when we refer to an approximation for
\textsc{Bin Packing} we implicitly mean the asymptotic approximation.
The first \emph{polynomial-time approximation scheme} (PTAS), that is, a $1+\eps$-approximation for any $\eps > 0$,
was given by Fernandez de la Vega and Lueker~\cite{fernandez81}.
Karmarkar and Karp~\cite{karmarkar82}
provided an algorithm which returns a solution with $\OPT + \O(\log^2 \OPT)$ bins.
Recently, Hoberg and Rothvo{\ss}~\cite{hoberg17_logarithmic} proved it is possible to find a solution
with $\OPT + \O(\log \OPT)$ bins in polynomial time.

The input for \textsc{Bin Packing} can be described by
$N$ numbers, corresponding to item sizes.
While in general these sizes may be distinct,
in some cases the input description can be compressed significantly 
by specifying the number of items of each size in the input.
Namely, in the \textsc{High-Multiplicity Bin Packing} problem, the input is
a set of pairs $(a_1, s_1), \dots, (a_\sigma, s_\sigma)$, where for $i=1,\dots,\sigma$, $a_i$ is
the number of items of size $s_i$ (and all $s_i$'s are distinct). 
Thus, $\sigma$ encodes the number of item sizes, and hence the size of
the description.
The goal is again to pack these items into bins, using as few bins as possible.
For constant number of sizes, $\sigma$,
Goemans and Rothvo{\ss}~\cite{goemans14_HMBP_polynomial} recently gave an exact algorithm
for the case of rational item sizes
running in time $\displaystyle{(\log \Delta)^{2^{\O(\sigma)}}}$, where $\Delta$ is the largest multiplicity of an item
or the largest denominator of an item size, whichever is the greater.

While these algorithms provide satisfying theoretical guarantees, 
simple heuristics are often adopted in practice to provide a
``good-enough'' performance.
\textsc{First Fit}~\cite{johnson74_fast_algs_bp}, which
puts each incoming item into the first bin where it fits and opens a new bin only
when the item does not fit anywhere else
achieves $1.7$-approximation~\cite{dosa13_FF_tight_analysis}.
For the high-multiplicity variant, using an LP-based Gilmore-Gomory cutting stock
heuristic~\cite{gilmore61_LP_cutting_stock,gilmore61_LP_cutting_stock2}
gives a good running time in practice~\cite{applegate03_cutting_stock}
and produces a solution with at most $\OPT + \sigma$ bins.
However, neither of these algorithms adapts well to the
streaming setting with possibly distinct item sizes.
For example, \textsc{First Fit} has to remember the remaining capacity
of each open bin, which in general can require space proportional to
$\OPT$. 

\textsc{Vector Bin Packing} proves to be substantially harder to approximate, even in a constant dimension.
For fixed $d$, Bansal, Eli\'{a}\v{s}, and Khan~\cite{bansal16_vector_BP} showed 
an approximation factor of $\approx 0.807 + \ln (d+1) + \eps$.
For general $d$, a relatively simple algorithm based on an LP
relaxation, due to Chekuri and
Khanna~\cite{chekuri04_multidim_packing},
remains the best known, with an approximation guarantee of $1+\eps d+
\O(\log  \oneovereps)$. 
The problem is APX-hard even for $d=2$~\cite{woeginger97_no_ptas_for_vector_BP}, and 
cannot be approximated within a factor better than $d^{1-\eps}$
for any fixed $\eps > 0$~\cite{christensen17_survey_multidim_BP} if $d$ is arbitrarily large.
Hence, our streaming $d+\eps$-approximation for \textsc{Vector Bin Packing}
asymptotically achieves the offline lower bound.

\paragraph{Sampling-based algorithms.}
Sublinear-time approximation schemes constitute a
model related to, but distinct from, streaming algorithms. 
Batu, Berenbrink, and Sohler~\cite{batu09_sublinear} provide an algorithm
that takes $\Otilde\left(\sqrt{N}\cdot \poly(\oneovereps)\right)$
weighted samples, meaning that the probability of sampling an item is proportional to its size.
It outputs an asymptotic $1+\eps$-approximation of $\OPT$.
If uniform samples are also available, then sampling $\Otilde\left(N^{1/3}\cdot \poly(\oneovereps)\right)$ items
is sufficient.
These results are tight, up to a $\poly(\oneovereps, \log N)$ factor.
Later, Beigel and Fu~\cite{beigel12_dense_sublinear} focused on uniform sampling of items, 
proving that $\widetilde{\Theta}(N / \SIZE)$ samples are sufficient and necessary,
where $\SIZE$ is the total size of all items.
Their approach implies a streaming approximation scheme
by uniform sampling of the substream of big items.
However, the space complexity in terms of $\oneovereps$ is not stated in the paper, 
but we calculate this to be 
$\Omega\left(\eps^{-c}\right)$ for a constant $c\ge 10$.
Moreover, $\Omega(\oneoverepssquared)$ samples are clearly needed to estimate
the number of items with size close to $1$. 
Note that our approach is deterministic and
substantially different than taking a random sample from the stream.

\paragraph{Online algorithms.}
Online and streaming algorithms are similar in the sense that they are required
to process items one by one.
However,
an online algorithm
must make all its decisions immediately --- it must fix the placement 
of each incoming item on arrival.\footnote{Relaxations which allow a
  limited amount of ``repacking'' have also been considered~\cite{feldkord18_dynamic_BP_repacking}.}
A streaming algorithm can postpone such decisions to the very end,
but is required to keep its memory small, whereas an online algorithm may remember 
all items that have arrived so far.
Hence, online algorithms apply in the streaming setting
only when they have small space cost, including the space needed to store the
solution constructed so far.
The approximation ratio of online algorithms is quantified by the
\emph{competitive ratio}.

For \textsc{Bin Packing}, the best possible competitive ratio is substantially
worse than what we can achieve offline or even in the streaming setting.
Balogh~\etal~\cite{balogh18_online_bp} designed 
an asymptotically $1.5783$-competitive algorithm, while
the current lower bound on the asymptotic competitive ratio is $1.5403$~\cite{balogh12_online_bp_LB}.
This (relatively complicated) online algorithm
is based on the \textsc{Harmonic} algorithm~\cite{lee85_harmonic},
which for some integer $K$ classifies items into size groups 
$(0,\frac1K], (\frac1K, \frac{1}{K-1}], \dots, (\half, 1]$.
It packs each group separately by \textsc{Next Fit}, 
keeping just one bin open, which is closed whenever the next item does not fit.
Thus \textsc{Harmonic} can run in memory of size $K$, unlike most other online algorithms
which require maintaining the levels of all bins opened so far.
Its competitive ratio tends to approximately $1.691$ as $K$ goes to infinity.
Surprisingly, this is also the best possible ratio if only a bounded
number of bins is allowed to be open for an online algorithm~\cite{lee85_harmonic},
which can be seen as the intersection of online and streaming algorithms.

For \textsc{Vector Bin Packing}, the best known competitive ratio
of $d+0.7$~\cite{garey76_resource_constr_sch_as_vbp} is achieved by \textsc{First Fit}.
A lower bound of $\Omega(d^{1-\eps})$ on the competitive ratio was
shown by Azar \etal~\cite{azar13_vbp_tight_bounds}.
It is thus currently unknown whether or not online algorithms outperform
streaming algorithms in the vector setting.

\subsection{Scheduling}

\paragraph{Offline approximation algorithms.}
\textsc{Makespan Scheduling} is strongly NP-complete~\cite{garey79_computers_intractability}, which 
in particular rules out the possibility of a PTAS with time complexity $\poly(\oneovereps, n)$.
After a sequence of improvements, Jansen, Klein, and Verschae~\cite{jansen16_eptas_scheduling}
gave a PTAS with time complexity $2^{\Otilde(1/\eps)} + \O(n \log n)$,
which is essentially tight under the Exponential Time Hypothesis
(ETH)~\cite{chen14_scheduling_LBs}.

For constant dimension $d$, \textsc{Vector Scheduling} also admits a PTAS, 
as shown by Chekuri and Khanna~\cite{chekuri04_multidim_packing}.
However, the running time is of order $n^{(1/\eps)^{\Otilde(d)}}$. 
The approximation scheme for a fixed $d$ was improved to an efficient
PTAS, namely to an algorithm running in time $2^{(1/\eps)^{\Otilde(d)}} + \O(dn)$, by Bansal~\etal~\cite{bansal16_apx_vector_sch_eptas},
who also showed that the running time cannot be significantly improved under ETH.
In contrast 
our streaming $\poly(d, m)$-space algorithm computes
an input summary maintaining $2$-approximation of the original input.
This respects the lower bound, since to compute the actual
makespan from the summary, we still need to execute an offline algorithm, with running time doubly exponential in $d$.
The  best known approximation ratio for large $d$ is
$\O(\log d / (\log \log d))$~\cite{Harris13_moser-tardos_framework,im19_tight_bounds_online_vector_sch},
while $\alpha$-approximation is not possible 
in polynomial time for any constant $\alpha > 1$ and arbitrary $d$, unless NP = ZPP.

\paragraph{Online algorithms.}
For the scalar problem, the optimal competitive ratio is known to lie in the interval
$(1.88, 1.9201)$~\cite{albers99_better_bounds_online_sch,gormley00_online_sch_LB,rudin01_phd_thesis,fleischer00_online_sched_algorithm},
which is substantially worse than what can be done by a simple
streaming $1+\eps$-approximation in space $\O(\oneovereps\cdot \log \oneovereps)$.
Interestingly, for \textsc{Vector Scheduling},
the algorithm by Im~\etal~\cite{im19_tight_bounds_online_vector_sch} with ratio $\O(\log d / (\log \log d))$
actually works in the online setting as well and needs space $\O(d\cdot m)$ only during its execution
(if the solution itself is not stored),
which makes it possible to implement it in the streaming setting.
This online ratio cannot be improved as there is a lower bound of
$\Omega(\log d / (\log \log d))$~\cite{im19_tight_bounds_online_vector_sch,azar18_randomized_vector_load_balancing},
whereas in the streaming setting we can achieve a $2$-approximation with a reasonable memory
(or even $1+\eps$ for a fixed $d$).
If all jobs have sufficiently small size, we improve the analysis
in~\cite{im19_tight_bounds_online_vector_sch} and show that the online algorithm
achieves $1+\eps$-approximation; see Section~\ref{sec:vsch}.

\section{Bin Packing}\label{sec:bp}

\paragraph{Notation.} 
For an instance $I$, let $N(I)$ be the number of items in $I$,
let $\SIZE(I)$ be the total size of all items in $I$, and
let $\OPT(I)$ be the number of bins used in an 
optimal solution for $I$. Clearly, $\SIZE(I) \le \OPT(I)$.
For a bin $B$, let $s(B)$ be the total size of items in $B$.
For a given $\eps > 0$, we use $\Otilde(f(\oneovereps))$
to hide factors logarithmic in $\oneovereps$ and $\OPT(I)$, i.e.,
to denote $\O\big(f(\oneovereps) \cdot \polylog \oneovereps \cdot \polylog \OPT(I)\big)$.

\paragraph{Overview.}
We first briefly describe the approximation scheme of Fernandez de la Vega and Lueker~\cite{fernandez81},
whose structure we follow in outline.
Let $I$ be an instance of \textsc{Bin Packing}.
Given a precision requirement $\eps > 0$, we say that an item is \textit{small}
if its size is at most $\eps$; otherwise, it is \textit{big}.
Note that there are at most $\oneovereps \SIZE(I)$ big items.
The rounding scheme in~\cite{fernandez81}, called ``linear grouping'', works as follows:
We sort the big items by size non-increasingly and
divide them into groups of $k = \lfloor \eps\cdot \SIZE(I) \rfloor$ items
(the first group thus contains the $k$ biggest items).
In each group, we round up the sizes of all items to the size of
the biggest item in that group.
It follows that the number of groups and thus the number of distinct
item sizes (after rounding) is bounded by $\lceil\oneoverepssquared\rceil$.
Let $\IR$ be the instance of \textsc{High-Multiplicity Bin Packing}
consisting of the big items with rounded sizes.
It can be shown that $\OPT(\IB)\le \OPT(\IR)\le (1 + \eps)\cdot \OPT(\IB)$, where $\IB$
is the set of big items in $I$
(we detail a similar argument in Section~\ref{sec:rounding}).
Due to the bounded number of distinct item sizes, we can find a
close-to-optimal solution for $\IR$ efficiently.
We then translate this solution into a packing for $\IB$ in the natural way. 
Finally, small items are filled greedily (e.g., by First Fit) and it can be shown
that the resulting complete solution for $I$ is a $1 + \O(\eps)$-approximation.

Karmarkar and Karp~\cite{karmarkar82} proposed an improved rounding scheme, called ``geometric grouping''.
It is based on the observation that item sizes close to $1$
should be approximated substantially better than item sizes close to $\eps$.
We present a version of such a rounding scheme in Section~\ref{sec:bp-betterRounding}.

Our algorithm follows a similar outline with two stages (rounding and finding a solution for the rounded instance), but
working in the streaming model brings two challenges:
First, in the
rounding stage, we need to process the stream of items and output
a rounded high-multiplicity instance with few item sizes that are not too small,
while keeping only a small number of items in the memory. 
Second, the rounding of big items needs to be done carefully so that not much space is ``wasted'',
since in the case when the total size of small items is relatively large, we argue
that our solution is close to optimal by showing that the bins are nearly full on average.

\paragraph{Input summary properties.}
More precisely, we fix some $\eps > 0$ that is used to control the approximation
guarantee.
During the first stage, our algorithm has one variable which
accumulates the total size of all small items in the input stream, i.e., those of size at most $\eps$.
Let $\IB$ be the substream consisting of all big items. 
We process $\IB$ and output a rounded high-multiplicity instance $\IR$
with the following properties:

\begin{enumerate}[nosep,label=(P\arabic*)]
\item \label{p:dSizes} There are at most $\sigma$ item sizes in $\IR$, 
all of them larger than $\varepsilon$,
and the memory required for processing $\IB$ is $\O(\sigma)$.
\item \label{p:roudingUp} The $i$-th biggest item in $\IR$ is at least as large
as the $i$-th biggest item in $\IB$ (and the number of items in $\IR$ is the same as in $\IB$).
This immediately implies that
\begin{itemize}[nosep]
\item Any packing of $\IR$ can be used as a packing of $\IB$ (in the same number of bins),
\item $\OPT(\IB)\le \OPT(\IR)$, and
\item $\SIZE(\IB) \le \SIZE(\IR)$.
\end{itemize}
\item \label{p:boundOnOPT} $\OPT(\IR)\le (1 + \eps)\cdot \OPT(\IB) + \O(\log \oneovereps)$.
\item \label{p:boundOnSIZE}$\SIZE(\IR) \le (1 + \eps)\cdot \SIZE(\IB)$.
\end{enumerate}

In words, \ref{p:roudingUp} means that we are rounding item sizes up
and, together with \ref{p:boundOnOPT}, it implies that the optimal
solution for the rounded instance
approximates $\OPT(\IB)$ well.
The last property is used in the case when the total size of 
small items constitutes a large fraction of the total size of all items.
Note that $\SIZE(\IR) - \SIZE(\IB)$ can be thought of as bin space ``wasted'' 
by rounding.

Observe that the succinctness of the rounded instance depends on 
$\sigma$. First, we show a streaming algorithm for rounding with
$\sigma= \Otilde(\oneoverepssquared)$. Then we improve upon it and
give an algorithm with $\sigma = \Otilde(\oneovereps)$, which is essentially
the best possible, while guaranteeing an error of $\eps\cdot \OPT(\IB)$
introduced by rounding (elaborated on in Section~\ref{sec:BPandQS}).
More precisely, we show the following:

\begin{lemma}\label{lem:rounding}
Given a steam $\IB$ of big items, there is a deterministic streaming algorithm that outputs
a \textsc{High-Multiplicity Bin Packing} instance satisfying \ref{p:dSizes}-\ref{p:boundOnSIZE}
with $\sigma = \O\left(\frac{1}{\eps}\cdot \log \frac{1}{\eps}\cdot \log \OPT(\IB)\right)$.
\end{lemma}

Before describing the rounding itself and proving Lemma~\ref{lem:rounding},
we explain how to use it to calculate an accurate estimate of the
number of bins.

\paragraph{Calculating a bound on the number of bins after rounding.}
First, we obtain a solution $\calS$ of the rounded instance $\IR$.
For instance, we may round the solution of the linear program introduced 
by Gilmore and Gomory~\cite{gilmore61_LP_cutting_stock,gilmore61_LP_cutting_stock2},
and get a solution with at most $\OPT(\IR) + \sigma$ bins.
Or, if item sizes are rational numbers, we may compute an optimal solution for $\IR$ by the algorithm
of Goemans and Rothvo{\ss}~\cite{goemans14_HMBP_polynomial}; however,
the former approach appears to be more efficient and more general.
In the following, we thus assume that $\calS$ uses at most $\OPT(\IR) + \sigma$ bins.

We now calculate a bound on the number of bins in the original instance.
Let $W$ be the total free space in the bins of $\calS$ that can be used for small items.
To be precise, $W$ equals the sum over all bins $B$ in $\calS$ of $\max(0, 1 - \eps - s(B))$.
Note that the capacity of bins is capped at $1 - \eps$, because
it may happen that all small items are of size $\eps$ while the packing
leaves space of just under $\eps$ in any bin.
Then we would not be able to pack small items into these bins.
Reducing the capacity by $\eps$ removes this issue. 
On the other hand, if a small item does not fit into a bin, then 
the remaining space in the bin is smaller than $\eps$.

Let $s$ be the total size of all small items in the input stream.
If $s \le W$, then all small items surely fit into the free space of bins in $\calS$
(and can be assigned there greedily by \textsc{First Fit}).
Consequently, we output that the number of
bins needed for the stream of items is at most $|\calS|$, i.e., the number
of bins in solution $\calS$ for $\IR$.
Otherwise, we need to place small items of total size at most $s' = s - W$
into new bins and it is easy to see that opening at most $\lceil s' / (1 - \eps)\rceil \le (1 + \O(\eps))\cdot s' + 1$ bins
for these small items suffices.
Hence, in the case $s > W$, 
we output that $|\calS| + \lceil s' / (1 - \eps)\rceil$ bins are sufficient to pack
all items in the stream.

We prove that the number of bins that we output
in either case is a good approximation of the optimal number of bins,
provided that $\calS$ is a good solution for $\IR$.

\begin{lemma}\label{lem:bp-estimate}
Let $I$ be given as a stream of items.
Suppose that $0<\eps \le \frac13$, that
the rounded instance $\IR$, created from $I$, satisfies properties \ref{p:dSizes}-\ref{p:boundOnSIZE},
and that the solution $\calS$ of $\IR$ uses at most $\OPT(\IR) + \sigma$ bins. 
Let $\ALG(I)$ be the number of bins that our algorithm outputs.
Then, it holds that
$\OPT(I)\le \ALG(I) \le (1 + 3\eps)\cdot \OPT(I) + \sigma + \O\left(\log \oneovereps\right)$.
\end{lemma}

\begin{proof}
We analyze the two cases of the algorithm:

\mycase{$s\le W$} In this case, small items fit into the bins of $\calS$
and $\ALG(I) = |\calS|$.
For the inequality $\OPT(I)\le \ALG(I)$, observe that the packing $\calS$
can be used as a packing of items in $\IB$ (in a straightforward way) with no less
free space for small items by property~\ref{p:roudingUp}. Thus $\OPT(I) \le |\calS|$.

To upper bound $\ALG(I)$, note that
$$|\calS| \le \OPT(\IR) + \sigma
\le (1 + \eps)\cdot \OPT(\IB) + \O\left(\log \oneovereps\right) + \sigma
\le (1 + \eps)\cdot \OPT(I) + \O\left(\log \oneovereps\right) + \sigma\,,$$
where the second inequality follows from property~\ref{p:boundOnOPT}
and the third inequality holds as $\IB$ is a subinstance of $I$.

\mycase{$s> W$} Recall that $\ALG(I) = |\calS| + \lceil s' / (1 - \eps)\rceil$.
We again have that $\calS$ can be used as a packing of $\IB$ with no less
free space for small items. Thus, the total size of small items that do not fit
into bins in $\calS$ is at most $s'$ and these items clearly fit into $\lceil s' / (1 - \eps)\rceil$ bins.
Hence, $\OPT(I)\le |\calS| + \lceil s' / (1 - \eps)\rceil$.

For the other inequality, consider starting with solution $\calS$ for $\IR$,
first to (almost) fill up the bins of $\calS$ with small items of total size $W$,
then using $\lceil s' / (1 - \eps)\rceil$ additional bins for the
remaining small items.
Note that in each bin, except the last one, the unused space is less than $\eps$,
thus the total size of items in $\IR$ and small items is more than $(\ALG(I) - 1) \cdot (1 - \eps)$.
Finally, we replace items in $\IR$ by items in $\IB$ and the total size of items
decreases by $\SIZE(\IR) - \SIZE(\IB)\le \eps\cdot \SIZE(\IB)\le \eps\cdot \SIZE(I)$ by property~\ref{p:boundOnSIZE}.
Hence, $\SIZE(I)\ge (\ALG(I) - 1) \cdot (1 - \eps) - \eps\cdot \SIZE(I)$.
Rearranging and using $\eps \le \frac13$, we get
$$\ALG(I)\le \frac{1 + \eps}{1 - \eps}\cdot \SIZE(I) + 1 \le (1 + 3\eps)\cdot \OPT(I) + 1\,.$$

Considered together, these two cases both meet the claimed bound. 
\end{proof}

\subsection{Processing the Stream and Rounding}
\label{sec:rounding}

The streaming algorithm of the rounding stage makes use of the deterministic quantile summary of
Greenwald and Khanna~\cite{greenwald01_quantile_summaries}. 
Given a precision $\delta > 0$ and an input stream of numbers $s_1, \dots, s_N$,
their algorithm computes a data structure $Q(\delta)$ which is able to answer
a quantile query with precision $\delta N$.
Namely, for any $0\le \phi\le 1$, it returns an element $s$ of the input stream
such that the rank of $s$ is $[(\phi-\delta)N, (\phi+\delta)N]$,
where the rank of $s$ is the position of $s$ in the non-increasing ordering of the input stream.%
\footnote{Note that if $s$ appears more times in the stream, its rank is an interval rather
than a single number. Also, unlike in~\cite{greenwald01_quantile_summaries}, we order numbers non-increasingly,
which is more convenient for \textsc{Bin Packing}.}
The data structure stores an ordered sequence of tuples, each consisting of an input 
number $s_i$ and valid lower and upper bounds on the true rank of $s_i$
in the input sequence.\footnote{More precisely, valid lower and upper bounds
on the rank of $s_i$ can be computed easily from the set of tuples.}
The first and last stored items correspond to
the maximum and minimum numbers in the stream, respectively.
Note that the lower and upper bounds on the rank of any stored number differ by at most $\lfloor 2\delta N\rfloor$
and upper (or lower) bounds on the rank of two consecutive stored numbers differ by at most $\lfloor 2\delta N\rfloor$ as well.
The space requirement of $Q(\delta)$ is $\O(\frac{1}{\delta}\cdot \log \delta N)$, however,
in practice the space used is observed to scale linearly with $\frac{1}{\delta}$~\cite{luo16_quantiles_experimental}.
(Note that an offline optimal data structure for $\delta$-approximate quantiles uses space $\O\left(\frac{1}{\delta}\right)$.)
We use data structure $Q(\delta)$ to construct our algorithm for processing the stream $\IB$ of big items.

\subsubsection{Simple Rounding Algorithm}
We begin by describing a simpler solution with $\delta = \frac14\eps^2$, resulting
in a rounded instance with $\Otilde(\oneoverepssquared)$ item sizes.
Subsequently, we introduce a more involved solution with smaller space
cost.
The algorithm uses a quantile summary structure to determine the
rounding scheme. 
Given a (big) item $s_i$ from the input, we insert it into $Q(\delta)$.
After processing all items, we extract from $Q(\delta)$ the set of stored
input items (i.e., their sizes) together with upper bounds on their
rank (where the largest size has highest rank 1, and the smallest size
has least rank $N$). 
Note that the number $\NB$ of big items in $\IB$ is less than
$\oneovereps\SIZE(\IB)\le \oneovereps\OPT(\IB)$ as each is of size more than $\eps$.
Let $q$ be the number of items (or tuples) extracted from $Q(\delta)$;
we get that $q = \O(\frac{1}{\delta}\cdot \log \delta \NB) = \O\big(\oneoverepssquared\cdot \log (\eps\cdot \OPT(\IB))\big)$.
Let $(a_1, u_1 = 1), (a_2, u_2), \dots, (a_q, u_q = \NB)$ be the output pairs of
an item size and the bound on its rank, sorted so that $a_1\ge a_2\ge \cdots \ge a_q$.
We define the rounded instance $\IR$
with at most $q$ item sizes as follows:
$\IR$ contains $(u_{j+1} - u_j)$ items of size $a_j$ for each $j = 1, \dots, q-1$,
plus one item of size $a_q$.
See Figure~\ref{fig:bp-rounding} for an illustration.

\begin{figure}
\centerline{\includegraphics[scale=1]{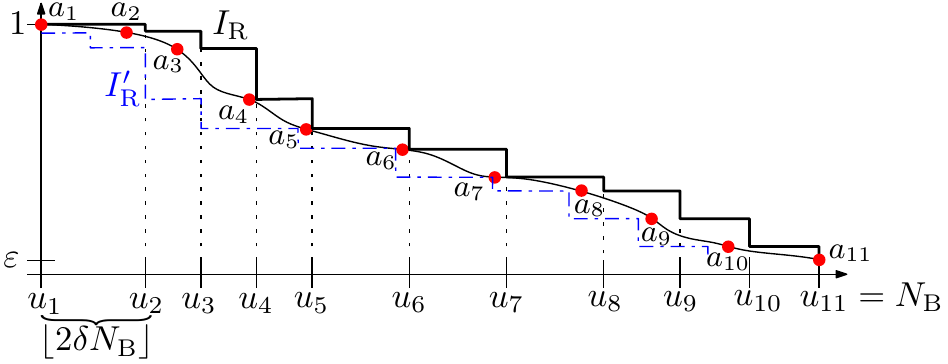}}
\caption{
An illustration of the original distribution of sizes of big items in $\IB$, depicted by a smooth curve,
and the distribution of item sizes in the rounded instance $\IR$, depicted by a bold ``staircase'' function.
The distribution of $\IR'$ (which is $\IR$ without the $\lfloor 4\delta \NB\rfloor$ biggest items)
is depicted a (blue) dash dotted line.
Selected items $a_i, \dots, a_q$, with $q=11$, are illustrated by (red) dots,
and the upper bounds $u_1, \dots, u_q$ on the ranks appear on the $x$ axis.
}
\label{fig:bp-rounding}
\end{figure}

We show that the desired properties \ref{p:dSizes}-\ref{p:boundOnSIZE} hold with $\sigma = q$.
Property \ref{p:dSizes} follows easily from the definition of $\IR$
and the design of data structure $Q(\delta)$.
Note that the number of items is preserved.
To show \ref{p:roudingUp}, suppose for a contradiction that the $i$-th biggest item in $\IB$
is bigger than the $i$-th biggest item in $\IR$, whose size is $a_j$ for $j=1,\dots,q-1$, i.e.,
$i\in [u_j, u_{j+1})$ (note that $j < q$ as $a_q$ is the smallest item in $\IB$ and is present only once in $\IR$).
We get that the rank of item $a_j$ in $\IB$ is strictly more than $i$,
and as $i \ge u_j$, we get a contradiction with the fact that $u_j$ is a valid upper bound
on the rank of $a_j$ in $\IB$.

Next, we give bounds for $\OPT(\IR)$ and $\SIZE(\IR)$, which are required by properties~\ref{p:boundOnOPT}
and~\ref{p:boundOnSIZE}.
We pack the $\lfloor 4\delta \NB\rfloor$ biggest items in $\IR$ separately into ``extra'' bins.
Using the choice of $\delta = \frac14\eps^2$ and $\NB\le \oneovereps\SIZE(\IB)$,
we bound the number of these items and thus extra bins by
$4\delta \NB\le \eps\cdot \SIZE(\IB)\le \eps\cdot \OPT(\IB)$.
Let $\IR'$ be the remaining items in $\IR$.
We claim that that the $i$-th biggest item $b_i$ in $\IB$ is bigger than
the $i$-th biggest item in $\IR'$ with size equal to $a_j$ for $j=1,\dots,q$.
For a contradiction, suppose that $b_i < a_j$,
which implies that the rank $r_j$ of $a_j$ in $\IB$ is less than $i$.
Note that $j < q$ as $a_q$ is the smallest item in $\IB$.
Since we packed the $\lfloor 4\delta \NB\rfloor$ biggest items from $\IR$ separately, 
one of the positions of $a_j$ in the ordering of $\IR$ is
$i + \lfloor 4\delta \NB\rfloor$ and so we have $i + \lfloor 4\delta \NB\rfloor < u_{j+1} \le u_j + \lfloor 2\delta \NB\rfloor$,
where the first inequality holds by the construction of $\IR$
and the second inequality is by the design of data structure $Q(\delta)$.
It follows that $i < u_j - \lfloor 2\delta \NB\rfloor$.
Combining this with $r_j < i$, we obtain that the rank of $a_j$ in $\IB$ is less than $u_j - \lfloor 2\delta \NB\rfloor$,
which contradicts that $u_j - \lfloor 2\delta \NB\rfloor$ is a valid
lower bound on the rank of $a_j$.

The claim implies $\OPT(\IR') \le \OPT(\IB)$ and $\SIZE(\IR') \le \SIZE(\IB)$.
We thus get that $\OPT(\IR) \le \OPT(\IR') + \lfloor 4\delta \NB\rfloor \le \OPT(\IB) + \eps\cdot \OPT(\IB)$,
proving property~\ref{p:boundOnOPT}.
Similarly, $\SIZE(\IR) \le \SIZE(\IR') + \lfloor 4\delta \NB\rfloor \le \SIZE(\IB) + \eps\cdot \SIZE(\IB)$,
showing \ref{p:boundOnSIZE} and concluding the analysis of our simple rounding algorithm.

\subsubsection{Improved Rounding Algorithm}
\label{sec:bp-betterRounding}

Our improved rounding algorithm reduces the number of sizes in the rounded instance
(and also the memory requirement) from $\Otilde(\oneoverepssquared)$ to $\Otilde(\oneovereps)$.
It is based on the observation that the number of items of sizes close
to $\eps$ can be approximated with much lower accuracy than the
number of items with sizes close to 1, without affecting the quality
of the overall approximation. This was observed already by Karmarkar and Karp~\cite{karmarkar82}.

\begin{proof}[Proof of Lemma~\ref{lem:rounding}]
Let $k = \lceil \log_2 \oneovereps \rceil$. 
We first group big items in $k$ groups $0, \dots, k-1$ by size such that in group $j$
there are items with sizes in $(2^{-j-1}, 2^{-j}]$.
That is, the size intervals for groups are $(0.5, 1], (0.25, 0.5]$, etc.
Let $N_j$, $j = 0, \dots, k-1$, be the number of big items in group $j$;
clearly, $N_j < 2^{j+1} \SIZE(\IB)\le 2^{j+1} \OPT(\IB)$.
Note that the total size of items in group $j$ is in $(2^{-j-1}\cdot N_j, 2^{-j}\cdot N_j]$.
Summing over all groups, we get in particular that
\begin{equation}
\SIZE(\IB) > \sum_{j=0}^k \frac{N_j}{2^{j+1}}\,.
\label{eqn:boundOnSIZE}
\end{equation}

For each group $j$, we use a separate data structure $Q_j := Q(\delta)$ with $\delta = \frac18\eps$,
where $Q(\delta)$ is the quantile summary from~\cite{greenwald01_quantile_summaries} with precision $\delta$.
So when a big item of size $s_i$ arrives,
we find $j$ such that $s_i\in (2^{-j-1}, 2^{-j}]$ and insert $s_i$ into $Q_j$. 
After processing all items, for each group $j$, we do the following:
We extract from $Q_j$ the set of stored
input items (i.e., their sizes) together with upper bounds on their rank.
Let $(a^j_1, u^j_1 = 1), (a^j_2, u^j_2), \dots, (a^j_{q_j}, u^j_{q_j} = N_j)$
be the pairs of an item size and the upper bound
on its rank in group $j$, ordered as in the simpler algorithm so that $a^j_1\ge a^j_2\ge \cdots \ge a^j_{q_j}$.
We have
$$q_j = \O\left(\frac{1}{\delta}\cdot \log \delta N_j\right)
      = \O\left(\frac{1}{\eps}\cdot \log \left(\eps \cdot 2^{j+1} \OPT(\IB)\right)\right)
      = \O\left(\frac{1}{\eps}\cdot \log \OPT(\IB)\right)
      ,$$
since $\eps 2^j \le \eps 2^k \le 1$. 
      
An auxiliary instance $\IR^j$ is formed by $(u^j_{i+1} - u^j_i)$ items of size $a_i$ for $i=1, \dots, q_j - 1$
plus one item of size $a_{q_j}$.
To create the rounded instance $\IR$, we take the union of all auxiliary instances $\IR^j$, $j = 0, \dots, k-1$.
Note that the number of item sizes in $\IR$ is
$$\sigma\le \sum_{j=0}^{k-1} q_j = \sum_{j=0}^{k-1} \O\left(\frac{1}{\eps}\cdot \log \OPT(\IB)\right)
= \O\left(\frac{k}{\eps}\cdot  \log \OPT(\IB)\right) 
= \O\left(\oneovereps\cdot \log \oneovereps\cdot  \log \OPT(\IB)\right) .$$

We show that the desired properties \ref{p:dSizes}-\ref{p:boundOnSIZE} are satisfied.
Property \ref{p:dSizes} follows easily from the definition of $\IR$ as the union of instances $\IR^j$
and the design of data structures $Q_j$.
To see property \ref{p:roudingUp}, for every group $j$, it holds that the $i$-th biggest item in group $j$ in $\IR$
is at least as large as the $i$-th biggest item in group $j$ in $\IB$.
Indeed, for any $p=0,\dots,q_j$,
$u^j_p$ is a valid upper bound on the rank of $a^j_p$ in group $j$ in $\IB$
and ranks of items of size $a^j_p$ in group $j$ in $\IR$ are at least $u^j_p$.
Moreover, the number of items is preserved in every group.
Hence, overall, the $i$-th biggest item in $\IR$
cannot be smaller than the $i$-th biggest item in $\IB$.

Next, we prove properties \ref{p:boundOnOPT} and \ref{p:boundOnSIZE},
i.e., the bounds on $\OPT(\IR)$ and on $\SIZE(\IR)$.
For each group $j$, we pack the $\lfloor 4\delta N_j\rfloor$ biggest items in $\IR$
with size in group $j$ into ``extra'' bins, each containing $2^j$ items, 
except for at most one extra bin which may contain fewer than $2^j$ items.
This is possible as any item in group $j$ has size at most $2^{-j}$.
Using the choice of $\delta = \frac18\eps$ and~\eqref{eqn:boundOnSIZE},
we bound the total number of extra bins by
\begin{equation}
\sum_{j=0}^k \left\lceil\frac{4\delta N_j}{2^j}\right\rceil
\le 4\cdot \frac18\eps \cdot  \sum_{j=0}^k \frac{N_j}{2^j} + k
\le \half\eps\cdot 2\cdot \SIZE(\IB) + k
\le \eps\cdot \OPT(\IB) + k\,.
\label{eqn:boundOnExtraBins}
\end{equation}

Let $\IR'$ be the remaining items in $\IR$.
Consider group $j$ and let $\IB(j)$ and $\IR'(j)$ be the items with sizes in 
$(2^{-j-1}, 2^{-j}]$ in $\IB$ and in $\IR'$, respectively.
We claim that the $i$-th biggest item $b_i$ in $\IB(j)$ is at least as large
as the $i$-th biggest item in $\IR'(j)$  with size equal to $a_p$ for $p=1,\dots,q_j$.
For a contradiction, suppose that $b_i < a_p$,
which implies that the rank $r_p$ of $a_p$ in $\IB(j)$ is less than $i$.
Note that $p < q_j$ as $a_{q_j}$ is the smallest item in $\IB(j)$.
Since we packed the largest $\lfloor 4\delta N_j\rfloor$ items from 
$\IR(j)$ separately, we have $i + \lfloor 4\delta N_j\rfloor < u_{p+1} \le u_p + \lfloor 2\delta N_j\rfloor$,
where the last inequality is by the design of data structure $Q_j$.
It follows that $i < u_p - \lfloor 2\delta N_j\rfloor$. Combining
it with $r_p < i$, we obtain that the rank of $a_p$ in $\IB(j)$ is less than $u_p - \lfloor 2\delta N_j\rfloor$,
which contradicts that $u_p - \lfloor 2\delta N_j\rfloor$ is a valid lower bound on the rank 
of $a_p$.
Hence, the claim holds for any group
and it immediately implies $\OPT(\IR') \le \OPT(\IB)$ and $\SIZE(\IR') \le \SIZE(\IB)$.

Combining with \eqref{eqn:boundOnExtraBins}, we get that
$\OPT(\IR) \le \OPT(\IR') + \eps\cdot \OPT(\IB) + k
	\le (1 + \eps)\cdot \OPT(\IB) + k,$
thus \ref{p:boundOnOPT} holds.
Similarly, to bound the total wasted space,
observe that the total size of items of $\IR$ that are not in $\IR'$
is bounded by
$$\sum_{j = 0}^k \frac{4\delta N_j}{2^j} \le 4\cdot \frac18\eps\cdot 2\cdot \sum_{j = 0}^k \frac{N_j}{2^{j+1}}
\le \eps\cdot \SIZE(\IB)\,,$$
where we use \eqref{eqn:boundOnSIZE} in the last inequality.
We obtain that
$\SIZE(\IR) \le \SIZE(\IR') + \eps\cdot \SIZE(\IB)
\le (1 + \eps)\cdot \SIZE(\IB).$
We conclude that properties \ref{p:dSizes}-\ref{p:boundOnSIZE} hold for the
rounded instance $\IR$.
\end{proof}

\subsection{Bin Packing and Quantile Summaries}\label{sec:BPandQS}

In the previous section, the deterministic quantile summary data structure from~\cite{greenwald01_quantile_summaries}
allows us to obtain a streaming approximation scheme for \textsc{Bin Packing}.
We argue that this connection runs deeper.

We start with different scenarios for which there exist better quantile summaries.
First, if all big item sizes belong to a universe $U\subset (\eps,1]$, then it can be better to use the quantile summary of Shrivastava~\etal~\cite{shrivastava04_bounded_universe_qs},
which provides a guarantee of $\O(\frac1\delta\cdot \log |U|)$ on the space complexity,
where $\delta$ is the precision requirement.
Thus, by using $k$ copies of this quantile summary in the same way as in Section~\ref{sec:bp-betterRounding},
we get a streaming $1+\eps$-approximation algorithm for \textsc{Bin Packing}
that runs in space $\O(\oneovereps\cdot \log \oneovereps \cdot \log |U|)$.

Second, if we allow the algorithm to use randomization and fail with probability $\gamma$,
we can employ the optimal randomized quantile summary of Karnin, Lang, and Liberty~\cite{karnin16_optimal_rand_quantile_summaries},
which, for a given precision $\delta$ and failure probability $\eta$, uses
space $\O(\frac1\delta\cdot \log\log \frac1\eta)$ and does not provide a $\delta$-approximate
quantile for some quantile query with probability at most $\eta$. 
In particular, using $k$ copies of their data structure with precision $\delta = \Theta(\eps)$
and failure probability $\eta = \gamma / k$ in the same way as in Section~\ref{sec:bp-betterRounding}
gives a streaming $1+\eps$-approximation algorithm for \textsc{Bin Packing}
which fails with probability at most $\gamma$
and runs in space $\O\left(\oneovereps\cdot \log \oneovereps \cdot \log\log (\log \oneovereps / \gamma)\right)$.

More intriguingly, the connection between quantile summaries and \textsc{Bin Packing}
also goes in the other direction.
Namely, we show that a streaming $1+\eps$-approximation algorithm
for \textsc{Bin Packing} with space bounded by $S(\eps,\OPT)$ (or $S(\eps, N)$)
implies a data structure of size $S(\eps,N)$ for the following \textsc{Estimating Rank} problem:
Create a summary of a stream of $N$ numbers which is able to
provide a $\delta$-approximate rank of any query $q$, i.e.,
the number of items in the stream which are larger than $q$, up to an additive error of $\pm \delta N$.
Observe that a summary for \textsc{Estimating Rank} is essentially a
quantile summary and we can actually use it to find an approximate
quantile by doing a binary search over possible item names.
However, this approach does \emph{not} guarantee that the item name returned will
correspond to one of the items present in the stream.

The reduction from \textsc{Estimating Rank} to \textsc{Bin Packing} goes as follows:
Suppose that all numbers in the input stream for \textsc{Estimating
  Rank} are from interval $(\frac12, \frac23)$
(this is without loss of generality by scaling) and let $q$ be a query in $(\frac12, \frac23)$.
For each such $a_i$ (in the \textsc{Estimating Rank} instance),
we introduce two items of size $a_i$ (in the \textsc{Bin Packing} instance). 
In the stream for \textsc{Bin Packing}, after the $2N$ items (two
copies each of $a_1, \dots, a_N$)
are inserted in the same order as in the stream for \textsc{Estimating Rank}, we then insert a further
$2N$ items, all of size $1-q$.
Observe first that no pair of the first $2N$ items can be placed in the
same bin, so we must open at least $2N$ bins, two for each of $a_1,
\ldots, a_N$.
Since $\frac12 > (1-q) > \frac13$, and $a_i > \frac12$, we can place at most
one of the $2N$ items of size $(1-q)$ in a bin with $a_i$ in it,
provided $a_i + (1-q) \leq 1$, i.e. $a_i \leq q$.
Thus, we can pack a number of the $(1-q)$-sized items, equivalent to
$2(N-\operatorname{rank}(q))$, in the first $2N$ bins.
This leaves $2\operatorname{rank}(q)$ items, all of size $(1-q)$.
We pack these optimally into $\operatorname{rank}(q)$ additional
bins, for a total of $2N + \operatorname{rank}(q)$ bins. 

We claim that a $1+\eps$-approximation of the optimum number of bins provides a 
$4\eps$-approximate rank of $q$.
Indeed, let $m$ be the number of bins returned by the algorithm
and let $r = m - 2N$ be the estimate of $\operatorname{rank}(q)$.
We have that the optimal number of bins equals $2N + \operatorname{rank}(q)$
and thus $2N + \operatorname{rank}(q) \le m \le (1+\eps)\cdot (2N + \operatorname{rank}(q)) + o(N)$.
By using $r = m - 2N$ and rearranging, we get
$$\operatorname{rank}(q) \le r \le \operatorname{rank}(q) + \eps \operatorname{rank}(q) + 2\eps N + o(N)\,.$$
Since the right-hand side can be upper bounded by $\operatorname{rank}(q) + 4\eps N$
(provided that $o(N) < \eps N$),
$r$ is a $4\eps$-approximate rank of $q$.
Hence, the memory state of an algorithm for \textsc{Bin Packing} after processing the first $2N$
items (of sizes $a_1, \dots, a_N$) can be used as a data structure for \textsc{Estimating Rank}.

%Hung and Ting~\cite{hung10_qs_lower_bound} show a 
%space lower bound of $\Omega(\oneovereps\cdot \log\oneovereps)$ for comparison-based quantile summaries.
In~\cite{cormode19_qs_lower_bound} we show a space lower bound of $\Omega(\oneovereps\cdot \log \eps N)$ for comparison-based 
data structures for \textsc{Estimating Rank} (and for quantile summaries as well).

\begin{theorem}[Theorem 13 in~\cite{cormode19_qs_lower_bound}]\label{thm:estimatingRank}
For any $0 < \eps < \frac1{16}$, 
there is no deterministic comparison-based data structure for \textsc{Estimating Rank}
which stores $o\left(\oneovereps\cdot \log \eps N\right)$ items on any input stream of length $N$.
\end{theorem}

We conclude that there is no comparison-based streaming algorithm for \textsc{Bin Packing}
which stores $o(\oneovereps\cdot \log \OPT)$ items on any input stream
(recall that $N = \O(\OPT)$ in our reduction).
Note that our algorithm is comparison-based if we employ the
comparison-based quantile summary of Greenwald and
Khanna~\cite{greenwald01_quantile_summaries},
except that it needs to determine the size group for each item, which
can be done by comparisons with $2^{-j}$ for integer values of $j$.
Nevertheless, comparisons with a fixed set of constants does not affect
the reduction from \textsc{Estimating Rank} (i.e., the reduction can choose an interval
to avoid all constants fixed in the algorithm), thus the lower bound of $\Omega\left(\oneovereps\cdot \log \eps N\right)$
applies to our algorithm as well.
This yields near optimality of our approach, up to a factor of $\O\left(\log \oneovereps\right)$.
Finally, we remark that the lower bound of $\Omega(\oneovereps\cdot \log \log \frac{1}{\delta})$
for randomized comparison-based quantiles summaries~\cite{karnin16_optimal_rand_quantile_summaries}
translates to \textsc{Bin Packing} as well.

\subsection{Vector Bin Packing}\label{app:vbp}

As already observed by Fernandez de la Vega and Lueker~\cite{fernandez81},
a $1+\eps$-approximation algorithm for (scalar) \textsc{Bin Packing} implies a
$d\cdot (1+\eps)$-approximation algorithm for \textsc{Vector Bin Packing}, where
items are $d$-dimensional vectors and bins have capacity $d$ in every dimension.
Indeed, we split the vectors into $d$ groups according to the largest dimension
(chosen arbitrarily among dimensions that have the largest value) and in each group 
we apply the approximation scheme for \textsc{Bin Packing}, packing just according
to the largest dimension. Finally, we take the union of opened bins over all groups.
Since the optimum of the \textsc{Bin Packing} instance
for each group is a lower bound on the optimum of \textsc{Vector Bin Packing},
we get that that the solution is a $d\cdot (1+\eps)$-approximation.

This can be done in the same way in the streaming model. Hence 
there is a streaming algorithm for \textsc{Vector Bin Packing} which outputs
a $d\cdot (1+\eps)$-approximation of $\OPT$, the offline optimal number of bins, using 
$\O\left(\frac{d}{\eps} \cdot \log \frac{1}{\eps} \cdot \log \OPT\right)$ memory.
By scaling $\eps$, there is a $d+\eps$-approximation algorithm
with $\Otilde(\frac{d^2}{\eps})$ memory.
We can, however, do better by one factor of $d$.

\begin{theorem}
There is a streaming $d+\eps$-approximation for \textsc{Vector Bin Packing} algorithm that uses 
$\O\left(\frac{d}{\eps} \cdot \log \frac{d}{\eps} \cdot \log \OPT\right)$ memory.
\end{theorem}

\begin{proof}
Given an input stream $I$ of vectors, we create an input stream $I'$
for \textsc{Bin Packing} by replacing each vector $\mathbf{v}$ by a
single (scalar) item $a$ of size $\|\mathbf{v}\|_\infty$.
We use our streaming algorithm for \textsc{Bin Packing}
with precision $\delta = \frac{\eps}{d}$ which uses 
$\O\left(\frac{1}{\delta} \cdot \log \frac{1}{\delta} \cdot \log \OPT\right)$ memory
and returns a solution with at most $B = (1 + \delta)\cdot \OPT(I') + \Otilde(\frac{1}{\delta})$ scalar bins.
Clearly, $B$ bins are sufficient for the stream $I$ of vectors, since
in the solution for $I'$ we replace each item by the corresponding vector and obtain a valid
solution for $I$.

Finally, we show that $(1 + \delta)\cdot \OPT(I') + \O_\delta(1)\le (d+\eps)\cdot \OPT(I) + \Otilde(\frac{d}{\eps})$ for which
it is sufficient to prove that $\OPT(I')\le d\cdot \OPT(I)$ as $\delta = \frac{\eps}{d}$. 
Namely, from an optimal solution $\calS$ for $I$, we create a solution	
for $I'$ with at most $d\cdot \OPT(I)$ bins. For each bin $B$ in $\calS$, 
we split the vectors assigned to $B$ into $d$ groups according to the largest
dimension (chosen arbitrarily among those with the largest value)
and for each group $i$ we create bin $B_i$ with vectors in group $i$.
Then we just replace each vector $v$ by an item of size $\|v\|_\infty$
and obtain a valid solution for $I'$ with at most $d\cdot \OPT(I)$ bins.
\end{proof}

Interestingly, a better than $d$-approximation using sublinear memory,
which is rounding-based, is not possible,
due to the following result in~\cite{bansal16_vector_BP}.
(Note that the result requires that the numbers in the input vectors can take arbitrary 
values in $[0,1]$, i.e., vectors do not belong to a bounded universe.)

\begin{theorem}[Implied by the proof of Theorem~2.2 in~\cite{bansal16_vector_BP}]
Any algorithm for \textsc{Vector Bin Packing} that rounds up 
large coordinates of vectors to $o(N / d)$ types cannot achieve
better than $d$-approximation, where $N$ is the number of vectors.
\end{theorem}

It is an interesting open question whether or not we can design
a streaming $d+\eps$-approximation with $o(\frac{d}{\eps})$ memory
or even with $\Otilde\left(d + \oneovereps\right)$ memory.

\section{Vector Scheduling}\label{sec:vsch}

\subsection{Streaming Algorithm by Combining Small Jobs}
\label{sec:vsch-combiningAlg}

We provide a novel approach for creating an input summary for
\textsc{Vector Scheduling} (and hence also \textsc{Makespan
  Scheduling} by inclusion),
based on combining small items into containers,
which works well even in a large dimension.
Our streaming algorithm stores all big jobs and all containers,
created from small items, that are relatively big as well.
Thus, there is a bounded number of big jobs and containers,
and the space used is bounded as well.
We show that this simple summarization preserves the optimal makespan up
to a factor of $2-\frac{1}{m} + \eps$ for any $0<\eps\le 1$.
Take $m\ge 2$, since for $m=1$ there is a trivial streaming algorithm
that just sums up the vectors of all jobs to get the optimal makespan.
We assume that the algorithm knows (an upper bound on) $m$ in advance.

\paragraph{Algorithm description.}
For $0<\eps \le 1$ and $m\ge 2$, the algorithms works as follows: 
For each $k = 1, \dots, d$,
it keeps track of the total load of all jobs in dimension $k$, denoted $L_k$.
Note that the optimal makespan satisfies $\OPT\ge \max_k \frac1m\cdot L_k$.
Assume for simplicity that when a new job arrives, $\max_k \frac1m\cdot L_k = 1$;
if not, we rescale every quantity by this maximum.
Hence, the optimum makespan for jobs that arrived so far is at least one,
while $L_k \le m$ for any $k = 1, \dots, d$
(an alternative lower bound on $\OPT$ is the maximum $\ell_\infty$
norm of a job seen so far, but our algorithm does not use this).

Let $\gamma = \Theta\left(\eps^2 / \log \frac{d^2}{\eps}\right)$;
the constant hidden in $\Theta$ follows from the analysis below.
We also ensure that $\gamma \le \frac14\eps$.
We say that a job with vector $\mathbf{v}$ is \emph{big} if $\|\mathbf{v}\|_\infty > \gamma$;
otherwise it is \emph{small}.
The algorithm stores all big jobs (i.e., the full vector of each big job),
while it aggregates small jobs into containers, and does not store any small job directly.
A \emph{container} is simply a vector $\mathbf{c}$ that equals the sum of vectors for small jobs
assigned to this container, and we ensure that $\|\mathbf{c}\|_\infty \le 2\gamma$.
Furthermore, container $\mathbf{c}$ is closed if $\|\mathbf{c}\|_\infty > \gamma$,
otherwise, it is open.
As two open containers can be combined into one (open or closed) container,
we maintain only one open container. We execute 
a variant of the \textsc{Next Fit} algorithm to pack the containers, adding the incoming
small job into the open container, where it always fits
as any small vector $\mathbf{v}$ satisfies $\|\mathbf{v}\|_\infty \le \gamma$.
All containers are retained in the memory.

When a new small job arrives or when a big job becomes small,
we assign it in the open container. If this container becomes closed,
we open a new, empty one. 
Moreover, it may happen that a previously closed container becomes open
again. In this case, we combine open containers as long as we have at
least two of them. This completes the description of the algorithm.

For packing the containers, we may also use another algorithm, such as \textsc{First Fit},
which also packs small jobs into a closed container if it fits.
This may lead to having a lower number of containers in some cases.
However, our upper bound of $2 - \frac{1}{m} + \eps$
on the approximation factor lost by this summarization technique holds for
any algorithm for packing the containers, as long as they do not exceed
the capacity of containers, equal to $2\gamma$.
Moreover, the space bound holds if there is a bounded number of
containers with $\|\mathbf{c}\|_\infty \le \gamma$. 

\paragraph{Properties of the input summary.}
After all jobs are processed, we assume again that $\max_k \frac1m\cdot L_k = 1$,
which implies that $\OPT\ge 1$.
Since any big job and any closed container, each characterized by a vector $\mathbf{v}$, satisfy 
$\|\mathbf{v}\|_\infty > \gamma$, 
it holds that there are at most $\frac{1}{\gamma} \cdot d\cdot m$
big jobs and closed containers. As at most one container remains open
in the end and any job or container is described by $d$ numbers,
the space cost is $\O\left(\frac{1}{\gamma} \cdot d^2\cdot m\right) = 
\O\left(\oneoverepssquared \cdot d^2\cdot m\cdot \log \frac{d}{\eps}\right)$.

We now analyze the maximum
approximation factor that can be lost by this summarization.
Let $\IR$ be the resulting instance formed by big jobs and
containers with small items (i.e., the input summary), and let 
$I$ be the original instance, consisting of jobs in the input stream.
We show that $\OPT(\IR)$ and $\OPT(I)$ are close together, up to a factor of $2 - \frac{1}{m} + \eps$.
Note, however, that we still need to execute an offline algorithm
to get (an approximation of) $\OPT(\IR)$, which is not an explicit part of the summary.

The crucial part of the proof is to show that containers for small items
can be assigned to machines so that the loads of all machines are nearly balanced
in every dimension,
especially in the case when containers constitute a large fraction of the total
load of all jobs.
To capture this, let $L^{\mathrm{C}}_k$ be the total load of containers
in dimension $k$ (equal to the total load of small jobs).
Let $\IC\subseteq \IR$ be the instance consisting of all containers in $\IR$.
The following lemma establishes the approximation factor. 

\begin{lemma}
Supposing that $\max_k \frac1m\cdot L_k = 1$, the following holds:
\begin{enumerate}[nosep,label=(\roman*)]
\item There is a solution for instance $\IC$ 
with load at most $\max(\half, \frac{1}{m}\cdot L^{\mathrm{C}}_k) +2\eps+4\gamma$
in each dimension $k$ on every machine.
\item $\OPT(I)\le \OPT(\IR)\le \left(2 - \frac{1}{m} + 3\eps\right)\cdot \OPT(I)$.
\end{enumerate}
\end{lemma}

\begin{proof}
(i) We obtain the solution from the randomized online algorithm by
Im \etal~\cite{im19_tight_bounds_online_vector_sch}.
Although this algorithm has ratio $\O(\log d / \log \log d)$ on general instances,
we show that it behaves substantially better when jobs are small enough.
In a nutshell, this algorithm works by first assigning each job $j$ to a uniformly random machine $i$ and if the load of machine $i$ exceeds a certain threshold, then the job
is reassigned by \textsc{Greedy}.
The online \textsc{Greedy} algorithm works by assigning jobs one by one,
each to a machine so that the makespan increases as little as possible
(breaking ties arbitrarily).

Let $L'_k = \max(\half, \frac{1}{m}\cdot L^{\mathrm{C}}_k)$.
We assume that each machine has its capacity of $L'_k+2\eps+4\gamma$
in each dimension $k$ split into two parts:
The first part has capacity $L'_k+\eps+2\gamma$ in dimension $k$
for the containers assigned randomly, and the second part has
capacity $\eps+2\gamma$ in all dimensions for the containers assigned by \textsc{Greedy}.
Note that \textsc{Greedy} cares about the load in the second part only.

The algorithm assigns containers one by one as follows: For each container $\mathbf{c}$,
it first chooses a machine $i$ uniformly and independently at random.
If the load of the first part of machine $i$
already exceeds $L'_k+\eps$ in some dimension $k$, then $\mathbf{c}$ is passed to \textsc{Greedy},
which assigns it according to the loads in the second part.
Otherwise, the algorithm assigns $\mathbf{c}$ to machine $i$.

As each container $\mathbf{c}$ satisfies $\|\mathbf{c}\|_\infty \le 2\gamma$,
it holds that randomly assigned containers fit into capacity $L'_k+\eps+2\gamma$
in any dimension $k$ on any machine.
We show that the expected amount of containers assigned by \textsc{Greedy} is small enough
so that they fit into machines with capacity of $\eps+2\gamma$,
which in turn implies that there is a choice
of random bits for the assignment so that the capacity for \textsc{Greedy} is not exceeded.
The existence of a solution with capacity $L'_k+2\eps+4\gamma$ in each dimension $k$ will follow.

Consider a container $\mathbf{c}$ and let $i$ be the machine chosen randomly for $\mathbf{c}$.
We claim that for any dimension $k$, the load on machine $i$ in dimension $k$,
assigned before processing $\mathbf{c}$, exceeds $L'_k+\eps$ 
with probability of at most $\frac{\eps}{d^2}$.
To show the claim, we use the following Chernoff-Hoeffding bound:
\begin{fact}
  Let $X_1, \dots, X_n$ be independent binary random variables and let
$a_1, \dots, a_n$ be coefficients in $[0,1]$. Let $X = \sum_i a_i X_i$.
Then, for any $0<\delta\le 1$ and any $\mu \ge \E[X]$, it holds that
$\Pr[X > (1+\delta)\cdot \mu] \le \exp\left(-\frac13\cdot \eps^2\cdot \mu\right)$.
\end{fact}

We use this bound with variable $X_{\mathbf{c'}}$
for each vector $\mathbf{c'}$ assigned randomly before vector $\mathbf{c}$ and not reassigned by \textsc{Greedy}.
We have $X_{\mathbf{c'}} = 1$ if $\mathbf{c'}$ is assigned on machine $i$.
Let $a_{\mathbf{c'}} = \frac{1}{2\gamma}\cdot \mathbf{c'}_k \le 1$. 
Let $X = \sum_{\mathbf{c'}} a_{\mathbf{c'}} X_{\mathbf{c'}}$ be the random variable
equal to the load on machine $i$ in dimension $k$, scaled by $\frac{1}{2\gamma}$.
It holds that
$\E[X] \le \frac{1}{m}\cdot \frac{1}{2\gamma}\cdot L'_k\cdot m = \frac{1}{2\gamma}\cdot L'_k$,
since each container $\mathbf{c'}$ is assigned to machine $i$ with probability $\frac{1}{m}$
and $L'_k\cdot m$ is the upper bound on the total load of containers in dimension $k$.
Using the Chernoff-Hoeffding bound with $\mu = \frac{1}{2\gamma}\cdot L'_k$ and $\delta = \eps\le 1$,
we get that
$$\Pr[X > (1+\eps)\cdot \frac{1}{2\gamma}\cdot L'_k]
	\le \exp\left(-\frac13\cdot \eps^2\cdot \frac{1}{2\gamma}\cdot L'_k\right)\,.$$
Using $\gamma = \O\left(\eps^2 / \log \frac{d^2}{\eps}\right)$ and $L'_k\ge \half$,
we obtain
$$\exp\left(-\frac13\cdot \eps^2\cdot \frac{1}{2\gamma}\cdot L'_k\right) \le
\exp\left(-\Omega\left(\log \frac{d^2}{\eps}\right)\right) \le \frac{\eps}{d^2}\,,$$
where the last inequality holds for a suitable choice of the multiplicative constant
in the definition of $\gamma$.
This is sufficient to show the claim as $X >(1+\eps)\cdot \frac{1}{2\gamma}\cdot L'_k$ if and only if
the load on machine $i$ in dimension $k$, assigned randomly before $\mathbf{c}$,
exceeds $(1+\eps)\cdot L'_k$.

By the union bound, the claim implies that each container $\mathbf{c}$ is
reassigned by \textsc{Greedy} with probability at most $\frac{\eps}{d}$. 
Let $G$ be the random variable equal to the sum of the $\ell_1$ norms
(where $\|\mathbf{c}\|_1 = \sum_{k=1}^d \mathbf{c}_k$)
of containers assigned by \textsc{Greedy}.
Using the linearity of expectation and the claim,
we have
$$\E[G] \le \sum_{\mathbf{c}} \frac{\eps}{d}\cdot \|\mathbf{c}\|_1
  \le \frac{\eps}{d} \cdot m\cdot d = \eps\cdot m\,,$$
where the second inequality uses that the total load of containers in each dimension 
is at most $m$.
Let $\mu_{\mathbf{G}}$ be the makespan of the containers created by \textsc{Greedy}.
Observe that each machine has a dimension with load at least $\mu_{\mathbf{G}} - 2\gamma$.
Indeed, otherwise, if there is a machine $i$ with load less than
$\mu_{\mathbf{G}} - 2\gamma$ in all coordinates,
the last container $\mathbf{c}$ assigned by \textsc{Greedy}
that caused the increase of the makespan to $\mu_{\mathbf{G}}$
would be assigned to machine $i$, and the makespan
after assigning $\mathbf{c}$ would be smaller than $\mu_{\mathbf{G}}$
(using $\|\mathbf{c}\|_\infty \le 2\gamma$).
It follows that $\mu_{\mathbf{G}} - 2\gamma \le \frac{1}{m}\cdot G$
and, using $\E[G]\le \eps\cdot m$,
we get that $\E[\mu_{\mathbf{G}}] - 2\gamma\le \eps$.
This concludes the proof that~(i) holds.

\medskip
(ii) The first inequality is straightforward as any solution for $\IR$ can 
be used as a solution for $I$, just packing small items first in containers
that the algorithm created
and then the containers according to the solution for $\IR$.

To show the upper bound, we create a solution of $\IR$ of makespan
at most $\left(2 - \frac{1}{m} + 3\eps\right)\cdot \OPT(I)$ as follows:
We take an optimal solution $\calS_{\mathrm{B}}$ for instance $\IR\setminus\IC$, i.e.,
for big jobs only,
and combine it with solution $\calS_{\mathrm{C}}$ for containers from~(i),
to obtain a solution $\calS$ for $\IR$
(interestingly, the machine loads from $\calS_{\mathrm{B}}$
and $\calS_{\mathrm{C}}$ can be combined in an arbitrary way).
Let $\mu_k$ be the largest load assigned to a machine in dimension $k$ in solution $\calS_{\mathrm{B}}$;
we have $\mu_k\le \OPT(I)$.
Note that $L^{\mathrm{C}}_k \le m - \mu_k$, since the total load 
of big jobs and containers together is at most $m$.
Consider the load on machine $i$ in dimension $k$ in solution $\calS$.
If $\frac{1}{m}\cdot L^{\mathrm{C}}_k \ge \half$, then this load
is bounded by 
$\mu_k + \frac{1}{m}\cdot L^{\mathrm{C}}_k +2\eps+4\gamma
\le \mu_k + \frac{1}{m}\cdot (m - \mu_k) +3\eps
= \left(1 - \frac{1}{m}\right)\cdot \mu_k + 1+3\eps
\le \left(2 - \frac{1}{m} + 3\eps\right)\cdot \OPT(I),$
where the first inequality uses $L^{\mathrm{C}}_k \le m - \mu_k$
and $\gamma \le \frac14\eps$ (ensured by the definition of $\gamma$),
and the last inequality holds by $\mu_k\le \OPT(I)$.

Otherwise, $\frac{1}{m}\cdot L^{\mathrm{C}}_k < \half$,
in which case the load on machine $i$ in dimension $k$
is at most $\mu_k + \half +2\eps+4\gamma \le (1.5 + 3\eps)\cdot \OPT(I) \le (2 - \frac{1}{m} + 3\eps)\cdot \OPT(I)$,
using similar arguments as in the previous case and $m\ge 2$.
\end{proof}

It remains open whether or not the above algorithm with $\gamma = \Theta(\eps)$
also gives $(2-\frac{1}{m}+\eps)$-app\-ro\-xi\-ma\-tion, which would imply
a better space bound of $\O(\oneovereps\cdot d^2\cdot m)$.
On the other hand, we now give an example
showing that the approximation guarantee is at least $2-\frac{1}{m}$ for this approach.

\paragraph{Tight example for the algorithm.}
For any $m\ge 2$, we present an instance $I$ in $d = m+1$ dimensions
such that $\OPT(I) = 1$, but $\OPT(\IR) \ge 2-\frac{1}{m}$, where $\IR$
is the instance created by our algorithm.

Let $\gamma$ be as in the algorithm and assume for simplicity that $\frac{1}{\gamma}$
is an integer. 
First, $m$ big jobs with vectors $\mathbf{v^1}, \dots, \mathbf{v^m}$ arrive,
where $\mathbf{v^i}$ is a vector with dimensions $i$ and $m+1$ equal to $1$ 
and with zeros in the other dimensions (that is, $\mathbf{v^i_{\mathit{i}}} = 1$ and $\mathbf{v^i_{\mathit{m+1}}} = 1$,
while $\mathbf{v^i_{\mathit{k}}} = 0$ for $k\notin \{i, m+1\}$).
Then, small jobs arrive in groups of $d-1 = m$ jobs and there are $(m-1)\cdot \frac{1}{\gamma}$ groups.
Each group consists of items $(\gamma, 0, \dots, 0, 0), (0, \gamma, \dots, 0, 0),
\dots, (0, 0, \dots, \gamma, 0)$, i.e, for each $i = 1, \dots, d-1$,
it contains one item with value $\gamma$ in coordinate $i$ and with zeros in other dimensions.
The groups arrive one by one, with an arbitrary ordering inside the group.
Note, however, that these jobs with $\ell_\infty$ norm equal to $\gamma$ become small for the algorithm
only once the first job from the last group arrives
as they are compared to the total load in each dimension, which increases gradually.
When they become small,
the algorithm will combine each group into one container $(\gamma, \gamma, \dots, \gamma, 0)$,
which can be achieved by processing the jobs in their arrival order
and by having the last vector of the group larger by an infinitesimal amount
(we do not take these infinitesimals into account in further calculations).
Thus, $\IR$ consists of $m$ big jobs and $(m-1)\cdot \frac{1}{\gamma}$
containers $(\gamma, \gamma, \dots, \gamma, 0)$.

Observe that $\OPT(I) = 1$, since in the optimal solution, each machine $i$ is assigned big job $\mathbf{v^i}$
and $\frac{1}{\gamma}$ small jobs with $\gamma$ in dimension $k$ for each $k \in \{1, \dots, d-1\} \setminus \{i\}$.
Thus the load equals one on any machine and dimension.

We claim that $\OPT(\IR) \ge 2-\frac{1}{m}$. Indeed, only one big job
can be assigned on one machine, as all of them have value one in dimension $m+1$,
so each machine contains one big job. Observe that 
some machine gets at least $\frac{m-1}{m}\cdot \frac{1}{\gamma}$ containers
and thus, it has load of at least $2-\frac{1}{m}$ in one of the $d-1$ first dimensions,
which shows the claim.

Note that for this instance to show ratio $2-\frac{1}{m}$ it suffices
that the algorithms creates $(m-1)\cdot \frac{1}{\gamma}$ containers $(\gamma, \gamma, \dots, \gamma, 0)$.
This can be enforced for various greedy algorithms used for packing the small jobs into containers.
We conclude that we need a different approach for input summarization to get a ratio
below $2-\frac{1}{m}$.

\subsection{Rounding Algorithms for Constant Dimension}
\label{app:vsch-roudingAlgs}

\paragraph{Makespan Scheduling.}
We start by outlining a simple streaming algorithm for $d=1$ based on rounding.
Here, each job $j$ on input is characterized by its processing time $p_j$ only.
The algorithm uses the size of the largest job seen so far,
denoted $p_{\mathrm{max}}$, as a lower bound on the optimum makespan. 
This makes the rounding procedure (and hence, the input summary) oblivious of $m$, the number of machines,
which is in contrast with the algorithm in Section~\ref{sec:vsch}
that uses just the sum of job sizes as the lower bound.

The rounding works as follows: 
Let $q$ be an integer such that $p_{\mathrm{max}}\in ((1+\eps)^q, (1+\eps)^{q+1}]$, and
let $k = \lceil\log_{1+\eps} \frac{1}{\eps}\rceil = \O(\oneovereps \log \frac{1}{\eps})$.
A job is \emph{big} if its size exceeds $(1+\eps)^{q-k}$; note that 
any big job is larger than $\eps\cdot p_{\mathrm{max}} / (1+\eps)^2$.
All other jobs are \emph{small} and have size less than $\eps\cdot p_{\mathrm{max}}$.
The algorithm maintains one variable $s$ for the total size of all small jobs
and variables $L_i$, $i=q-k, \dots, q$, for the number of big jobs with size in $((1+\eps)^i, (1+\eps)^{i+1}]$
(note that this interval is \emph{not} scaled by $p_{\mathrm{max}}$, i.e.,
increasing $p_{\mathrm{max}}$ slightly does not move the intervals).

Maintaining these variables when a new job arrives can be done in a straightforward way.
In particular, when an increase of $p_{\mathrm{max}}$ causes that
$q$ increases (by $1$ or more as it is integral), 
we discard all variables $L_i$ that do not correspond to big jobs any more,
and account for previously big jobs that are now small in variable $s$.
However, as the size of these jobs was rounded to a power of $1+\eps$,
variable $s$ can differ from the exact total size of small jobs
by a factor of at most $1+\eps$. 

The created input summary, consisting of $\O(\oneovereps \log \oneovereps)$ variables $L_i$
and variable $s$, preserves the optimal value
up to a factor of $1+\O(\eps)$. This follows, since big jobs are stored with size rounded up to the nearest power
of $1+\eps$, and, although we just know the approximate total size of small jobs,
they can be taken into account similarly as when calculating a bound on the number of bins
in our algorithm for \textsc{Bin Packing}.

\paragraph{Vector Scheduling.}
We describe the rounding introduced by Bansal \etal~\cite{bansal16_apx_vector_sch_eptas},
which we can adjust into a streaming $1+\eps$-approximation for \textsc{Vector Scheduling} in constant dimension.
The downside of this approach is that it requires memory exceeding $\left(\frac{2}{\eps}\right)^d$,
which becomes unfeasible even for $\eps=1$ and $d$ being a relatively small constant.
Moreover, such an amount of memory may be needed also in the case of a small number of machines.

We first use the following lemma by Chekuri and Khanna~\cite{chekuri04_multidim_packing},
where $\delta = \frac{\eps}{d}$:

\begin{lemma}[Lemma~2.1 in~\cite{chekuri04_multidim_packing}]\label{lem:roundingSmallCoords}
Let $I$ be an instance of \textsc{Vector Scheduling}. Let $I'$
be a modified instance where we replace each vector $\mathbf{v}$
by vector $\mathbf{v'}$ as follows: For each $1\le i\le d$, if $\mathbf{v}_i > \delta \|\mathbf{v}\|_\infty$,
then $\mathbf{v'}_i = \mathbf{v}_i$; otherwise, $\mathbf{v'}_i = 0$. 
Let $\calS'$ be any solution for $I'$.
Then, if we replace each vector $\mathbf{v'}$ in $\calS'$ by its counterpart in $I$,
we get a solution of $I$ with makespan at most $1+\eps$ times the makespan of $\calS'$.
\end{lemma}

In the following, we assume that the algorithms receives vectors from instance $I'$,
created as in Lemma~\ref{lem:roundingSmallCoords}. 
Let $p_{\mathrm{max}}$ be the maximum $\ell_\infty$ norm over all vectors
that arrived so far; we use it as a lower bound on $\OPT$.
We again do not use the total volume in each dimension as a lower bound,
which makes the input summarization oblivious of $m$.
A job, characterized by vector $\mathbf{v}$,
is said to be \emph{big} if $\|\mathbf{v}\|_\infty > \delta\cdot p_{\mathrm{max}}$;
otherwise, $\mathbf{v}$ is \emph{small}.

We round all values in big jobs to the powers of $1+\eps$.
By Lemma~\ref{lem:roundingSmallCoords}, we have that
either $\mathbf{v}_k > \delta^2\cdot p_{\mathrm{max}}$ or $\mathbf{v}_k = 0$
for any big $\mathbf{v}$ and dimension $k$, thus there are 
$\left\lceil \log_{1+\eps} \frac{1}{\delta^2}\right\rceil^d =
\O\left(\left(\frac{2}{\eps} \log \frac{d}{\eps}\right)^d\right)$
types of big jobs at any time. 
We have one variable $L_\mathbf{t}$ counting the number of
jobs for each big type $\mathbf{t}$, where $\mathbf{t}$
is an integer vector consisting of the exponents, i.e.,
if $\mathbf{v}$ is a big vector of type $\mathbf{t}$,
then $\mathbf{v}_i\in \left((1+\eps)^{\mathbf{t}_i}, (1+\eps)^{\mathbf{t}_i+1}\right]$
(we set $\mathbf{t}_i = -\infty$ if $\mathbf{v}_i = 0$).
As in the 1-dimensional case, big types change over time,
when $p_{\mathrm{max}}$ (sufficiently) increases.

Note that small jobs cannot be rounded to powers of $1+\eps$ directly.
Instead, they are rounded relative to their $\ell_\infty$ norms.
More precisely, consider a small vector $\mathbf{v}$ and let $\gamma = \|\mathbf{v}\|_\infty$.  
For each dimension $k$, if $\mathbf{v}_k > 0$, let $\mathbf{t}_k\ge 0$ be the largest integer
such that $\mathbf{v}_k \le \gamma\cdot (1+\eps)^{-\mathbf{t}_i}$, and
if $\mathbf{v}_i = 0$, we set $\mathbf{t}_i$ to $\infty$.
Then $(\mathbf{t}_1, \dots, \mathbf{t}_d)$ is the type of small vector $\mathbf{v}$.
Observe that small types do not change over time and there are 
at most $\O\left(\left(\frac{1}{\eps} \log \frac{d}{\eps}\right)^d\right)$ of them.
For each small type $\mathbf{t}$, we have one variable $s_\mathbf{t}$ counting the sum of the $\ell_\infty$ norms of
all small jobs of that type.

The variables can be maintained in an online fashion. 
Namely, when $p_{\mathrm{max}}$ increases, the types for previously
big jobs that are now small are discarded, while the jobs that become small
are accounted for in small types. For each such former big type $\mathbf{t}$,
we compute the corresponding small type as follows:
Let $\delta = \|\mathbf{t}\|_\infty$ be the maximum value in $\mathbf{t}$ (which is not $-\infty$).
The corresponding small type $\mathbf{\hat{t}}$ has then
$\mathbf{\hat{t}}_i = \delta - \mathbf{t}_i$ if $\mathbf{t}_i \neq -\infty$,
and $\mathbf{\hat{t}}_i = \infty$ otherwise.
Then we increase $s_\mathbf{\hat{t}}$ by $L_\mathbf{t}\cdot (1+\eps)^{\delta+1}$.

There are two types of errors introduced due to maintaining variables 
in the streaming scenario and not offline, where we know the final value of $p_{\mathrm{max}}$
in advance. 
First, it may happen that a vector $\mathbf{v}$ that was big upon its arrival becomes small,
and the small type of $\mathbf{v}$
is different than the small type computed for the former big type of $\mathbf{v}$
(i.e., the small type of $\mathbf{v}$ with values rounded to powers of $1+\eps$).
Second, the sum of $\ell_\infty$ norms
of small vectors of a small type $\mathbf{t}$ is in $\big(s_\mathbf{t} / (1+\eps), s_\mathbf{t}]$,
and moreover, the error in some dimension $i$ with $\mathbf{t}_i > 0$ (i.e., not the
largest one for this type) may be of factor up to $(1+\eps)^2$, since
we may round such a dimension two times for some jobs.
Note, however, that by giving up a factor of $1+\O(\eps)$, we may 
disregard both issues.

The offline algorithm of Bansal \etal~\cite{bansal16_apx_vector_sch_eptas}
implies that such an input summary, consisting of variables for both small and big types,
is sufficient for computing $1+\eps$-approximation.

\bigskip

\paragraph{Acknowledgments.}
The work is supported by 
European Research Council grant ERC-2014-CoG 647557.
The authors wish to thank Michael Shekelyan for fruitful discussions.

\bibliographystyle{plain}
\bibliography{references-BP}

\end{document}